\documentclass[12pt,
 reprint,
superscriptaddress,
 amsmath,amssymb,
 aps,
 showkeys, 
]{revtex4-2}

\usepackage{graphicx}
\usepackage{dcolumn}
\usepackage{listings}
\usepackage{amsmath,amsthm, bbm}
\usepackage[export]{adjustbox}
\usepackage{bm}

\newcommand{\bra}[1]{\ensuremath{\left\langle#1\right|}}
\newcommand{\ket}[1]{\ensuremath{\left|#1\right\rangle}}

\newcommand{\ketbra}[2]{\ensuremath{\left|#1\rangle\langle#2\right|}}
\newtheorem{theorem}{Theorem}[subsubsection]
\newtheorem{lemma}[theorem]{Lemma}
\theoremstyle{definition}
\newtheorem{definition}{Definition}[section]
\newtheorem{conjecture}{Conjecture}
\usepackage{tikz}
\usepackage{xcolor}
\newtheorem{proposition}{Proposition}
\usepackage{xpatch}

\makeatletter
\xpatchcmd{\proof}{\topsep6\p@\@plus6\p@\relax}{}{}{}
\makeatother

\definecolor{codegreen}{rgb}{0,0.6,0}
\definecolor{codegray}{rgb}{0.5,0.5,0.5}
\definecolor{codepurple}{rgb}{0.58,0,0.82}
\definecolor{backcolour}{rgb}{0.95,0.95,0.92}

\lstdefinestyle{mystyle}{
    backgroundcolor=\color{backcolour},   
    commentstyle=\color{codegreen},
    keywordstyle=\color{magenta},
    numberstyle=\tiny\color{codegray},
    stringstyle=\color{codepurple},
    basicstyle=\ttfamily\footnotesize,
    breakatwhitespace=false,         
    breaklines=true,                 
    captionpos=b,                    
    keepspaces=true,                 
    numbers=left,                    
    numbersep=5pt,                  
    showspaces=false,                
    showstringspaces=false,
    showtabs=false,                  
    tabsize=2
}

\begin{document}

\preprint{APS/123-QED}

\title{Quantum walk mixing is faster than classical on periodic lattices\\}

\author{Shyam Dhamapurkar}
 \email{shyam18596@gmail.com }
\affiliation{Shenzhen Institute for Quantum Science and Engineering (SIQSE),
Southern University of Science and Technology, Shenzhen, P. R. China}

\author{Xiu-Hao Deng}
\email{ dengxh@sustech.edu.cn }
\affiliation{Shenzhen Institute for Quantum Science and Engineering (SIQSE),
Southern University of Science and Technology, Shenzhen, P. R. China}
\affiliation{International Quantum Academy (SIQA), and Shenzhen Branch,
Hefei National Laboratory, Futian District, Shenzhen, P. R. China}

\date{\today}

\begin{abstract} The quantum mixing time is a critical factor affecting the efficiency of quantum sampling and algorithm performance. It refers to the minimum time required for a quantum walk to approach its limiting distribution closely and has implications across the areas of quantum computation. This work focuses on the continuous time quantum walk mixing on a regular graph, evolving according to the unitary map $U = e^{i \Bar{A}t}$, where the Hamiltonian $\Bar{A}$ is the normalized adjacency matrix of the graph. In~[Physical Review A 76, 042306 (2007).], Richter previously showed that this walk mixes in time $O(n d\log{(d)\log{(1/\epsilon)}})$ with $O( \log{(d)\log{(1/\epsilon)}})$ intermediate measurements when the graph is the $d-$dimensional periodic lattice $\mathbb{Z}_{n}\times \mathbb{Z}_{n}\times \dots \times \mathbb{Z}_{n}$. We extend this analysis to the periodic lattice $\mathcal{L} = \mathbb{Z}_{n_1}\times \mathbb{Z}_{n_2}\times \dots \times \mathbb{Z}_{n_d}$, relaxing the assumption that $n_i$ are identical.
We provide two quantum walks on periodic lattices that achieve faster mixing compared to classical random walks: 1. A coordinate-wise quantum walk that mixes in $O\Big(\Big(\sum_{i=1}^{d} n_i \Big) \log{(d/\epsilon)}\Big)$ time with $O(d \log(d/\epsilon))$ measurements. 2. A continuous-time quantum walk with $O(\log(1/\epsilon))$ measurements that conjecturally mixes in $O(\sum_{i=1}^d n_i(\log(n_1))^2 \log(1/\epsilon))$ time. Our results demonstrate a quadratic speedup over the classical mixing time $O(dn_1^2 \log(d/\epsilon))$ on the generalized periodic lattice $\mathcal{L}$. We have provided analytical evidence and numerical simulations to support the conjectured faster mixing time of the continuous-time quantum walk algorithm.
Making progress towards proving the general conjecture that quantum walks on regular graphs mix in $O(\delta^{-1/2} \log(N) \log(1/\epsilon))$ time, where $\delta$ is the spectral gap and $N$ is the number of vertices.
 \end{abstract}

\maketitle

\noindent {\bf {\em General introduction.---}} 
Quantum walk-based algorithms have been extensively explored as a means to achieve computational advantages over classical random walk algorithms for a variety of problems(~\cite{venegas2012quantum, kadian2021quantum} for a review). Quantum walks are quantum analogues of classical random walks where the walker takes a step in a quantum superposition of positions over a graph or lattice. 
Algorithms based on quantum walks have been used to solve problems such as
element distinctness, triangle finding, subset finding, decision tree~(see, e.g.\ Refs.\cite{Ambainis_element_distinctness, Magniez05quantumalgorithms,Childs_Subsetfinding,Farhi_1998}). 
 Many experiments on quantum walks have been extensively studied with the NISQ (noisy intermediate-scale quantum) devices~\cite{yan2019strongly, wang2022large, gong2021quantum}. There is a recent study on search problems with an improvement over previous search algorithms with a universal approach~\cite{wang2023universal,xu2022robust}.

\begin{figure}[ht!]
        \centering
        \includegraphics[trim = 110 250 130 245, width=\linewidth, right]{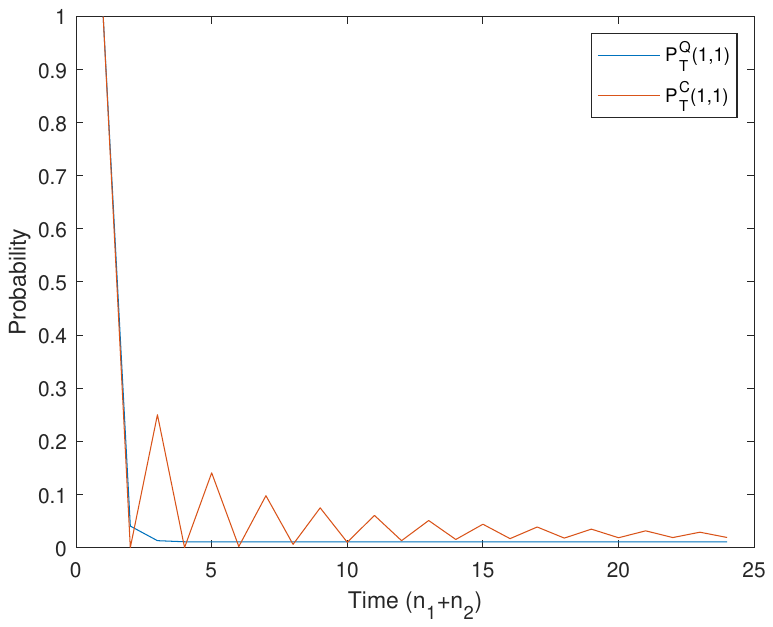}
        \caption{\textbf{Quantum walk vs Classical walk mixing on $\mathbb{Z}_{19} \times\mathbb{Z}_{5}$:} As an illustrative example, the time-averaged probability \textcolor{blue}{quantum} and \textcolor{orange}{classical} walk starting from the first vertex and measure at the same vertex after $T_{mix} = (n_1+n_2)$ steps. It is important to note that the quantum walk probability $P_T^Q(1,1)$ achieves a uniform probability value $1/n_1 n_2 = 0.01053$ in $T_{mix} = (n_1 + n_2) $ time, while the classical walk probability $P_T^C(1,1)$ only fluctuates around $1/n_1 n_2$. Classically, $T_{mix}$ is $(n_1^2 + n_2^2)$. This work proves the quadratic speedup in quantum walk mixing time on $\mathcal{L}$.
        }
        \label{fig:probability}
    \end{figure}

In this work, we are interested in quantum walks on the graphs and their mixing time. The mixing time of a (random or quantum) walk is the time for the walker's current distribution to become $\epsilon$ distance close to its limiting stationary distribution, regardless of the initial position. One of the famous applications is the Markov Chain Monte Carlo method for approximate sampling and counting~\cite{levin2017markov}. This work shows quadratic speedup in the quantum setting over the classical mixing time on $\mathcal{L}$. These results contribute to the ongoing effort of proving the conjecture that any quantum walk on a regular graph with $N$ vertices requires the square root of the classical mixing time steps to mix and reach a uniform distribution~\cite{Richter_2007}.

Quantum walks on graphs were considered by Richter~\cite{Richter_2007,PR1}, who established that they could be used to generate a near uniform distribution on the vertex set for various strategies for intermediate measurement that force decoherence. In particular, Richter demonstrated that by employing instantaneous and repeated randomized measurements, the continuous time quantum walks on $\mathbb{Z}_n^d$ achieves mixing in $O(nd \log{(d)}\log{(1/\epsilon)})$ time with $O(\log{(d)}\log{(1/\epsilon)})$ repeated measurements. The time complexity in this result is equivalent to $O(1/\sqrt{\delta}\log{1/\pi_{*}})$, where $\delta$ represents the spectral gap of the standard random walk matrix of $\mathbb{Z}_n^d$, and $\pi_{*}$ denotes the minimum among all entries of the stationary distribution $\pi$. Richter also proposed a conjecture~\cite{Richter_2007} that any quantum walk on regular graphs requires $O(\delta^{-1/2} \log(N) \log(1/\epsilon))$ time to mix. However, the conjecture remains unproven; our work is progressing towards validating this conjecture for a sub-class of regular graphs.  
 
We focus on the generalized periodic lattices $ \mathcal{L} = \mathbb{Z}_{n_1} \times \mathbb{Z}_{n_2} \times \dots \times \mathbb{Z}_{n_d}$, $n_1, n_2, \dots, n_d \in \mathbb{N}$. We note that Richter's analysis does not cover the case when the $n_i$ are not all equal. In the classical setting, we observe that the standard random walk on $\mathcal{L}$ reaches close to the stationary distribution in $O(dn_{1}^2\log(1/\epsilon))$ time, where $n_1 = \max_i\{n_i \vert 1 \leq i \leq d\}$. In a quantum setting, we do two types of quantum walks. In the \emph{first} type of quantum walk we evolve on each $\mathbb{Z}_{n_k}$ separately using the unitary $U = e^{iH_k t}$, where $H_k = I_{n_1} \otimes \dots \otimes \Bar{A}_k \otimes \dots \otimes I_{n_d}$ and measure in respective basis. We observe that this version of quantum walks on $\mathcal{L}$ mixes in $O\Big(\Big(\sum_{i=1}^{d} n_i \Big) \log{(d/\epsilon)}\Big)$ with $O(\log{(d/\epsilon)})$ measurements. Consequently, in the case where $n_j = n_1/j$ for $j \in [d]$, the improved mixing time bound is $O(n_1 (\log{d}) \log{(d/\epsilon)})$ in $O(d\log{(d/\epsilon)})$ measurements. The \emph{second} type is a continuous time quantum walk with repeated measurements~\cite{Childsexpo,childs2002example}. We provide a new direct analysis, and using a conjecture supported by numerical simulations, we show that the mixing time for the continuous time quantum walks on $\mathcal{L}$ for $d = 2$, i.e.~$\mathbb{Z}_{n_1} \times \mathbb{Z}_{n_2}$ is $O((n_1+ n_2)(\log{(n_1)})^2 \log{(1/\epsilon}))$ with $O(\log{(1/\epsilon)})$ measurements. Although this is slightly weaker than the result $(i)$ of the abstract for $d = 2$ case, the method gives a general approach to solving the mixing time problem for $\mathcal{L}$. In the end, we propose a hypothesis suggesting that the mixing time for a quantum walk on $\mathcal{L}$, with repeated measurement, is $O(\sum_{i}^d n_i (\log{(n_1)})^2 \log{(1/\epsilon)})$ with $O(\log{(1/\epsilon)})$ measurements. 

The structure of this paper is as follows: In the first section, we provide the necessary background information. Subsequently, prove that the classical random walks on $\mathcal{L}$ mixes within quadratic time. Moving on to the next section, we discuss the results using $O(d\log{(d/\epsilon)})$ measurements and provide proof for the general case. Then we present analytical evidence showcasing that for the special case of $\mathbb{Z}_{n_1} \times \mathbb{Z}_{n_2}$, the continuous time quantum walk mixes with repeated measurement in linear time. Then, we propose a hypothesis for the general periodic lattice problem. We conclude the paper with a summary and an outlook for future research.

\noindent {\bf {\em Preliminaries.---}} 
This section introduces key definitions and propositions related to Markov chains and mixing time~\cite{levin2017markov} since the random walk is a special case of a Markov chain. The section defines concepts such as Markov chain, stationary distribution, maximum pairwise column distance, and submultiplicativity, and presents two important propositions related to the mixing time of Markov chains~\cite{aldous1995reversible}.  

A Markov chain is a stochastic process ${X_0, X_1, X_2, ...}$ with a countable set of states $S$, where the probability of transitioning from one state to another depends only on the current state. Mathematically, for any states $i, j \in S$ and any time steps $t \geq 0$, the Markov property can be expressed as:

$P(X_{t+1} = j | X_0 = i, X_1 = x_1, \dots, X_t = x_t ) = P(X_{t+1} = j | X_t = i)$. $P(X_{t+1} = j | X_t = i)$ represents the probability of transitioning from state $i$ to state $j$ in a one-time step.

\begin{definition}
    Markov chain $P$ has a stationary distribution $\pi$ implies that $P \pi = \pi$.
\end{definition}
 
\begin{definition}
  Consider an irreducible (strongly connected) and aperiodic (non-bipartite) Markov Chain $P$ with a stationary distribution $\pi$. The mixing time(also known as \emph{threshold mixing}) can be defined as follows:

\begin{equation}
   \tau(1/2e) = min \Big \{ T: \frac{1}{2} \parallel P^{t} - \pi 1^\dagger \parallel_{1} \leq \frac{1}{2e} \hspace{1mm} \forall \hspace{1mm}t\geq T \Big \}, 
\end{equation}

where $\parallel . \parallel_{1}$ is a matrix 1-norm and $1^\dagger$ is all one row vector. Subsequently, $\tau(\epsilon)$ is called \emph{$\epsilon$-mixing}.
\end{definition}

Let $P$ be Markov chain, then 
\[d(P) = max_{jj'}\frac{1}{2}\parallel P(.,j) - P(.,j') \parallel_{1}\]
is called as maximum pairwise column distance. The following inequality holds for $d(P)$.

 \begin{equation}
    \frac{1}{2}\parallel P - \pi 1^\dagger \parallel_{1} \leq d(P) \leq \parallel P - \pi 1^\dagger \parallel_{1}.
 \end{equation}

 The distance $d(P)$ is submultiplicative, i.e. 

 \begin{equation}
 \label{submu}
  d(P_{t+t'}) \leq d(P_t) d(P_{t'})   
 \end{equation}
  for any time $t$ and $t'$. This implies that $d(P_t) \leq d(P)^t$ and $d((P_t)^{t'}) \leq d(P_t)^{t'}$.
 
 \begin{proposition}\label{fact1}\cite{aldous1995reversible}
 If $d(P) \leq \alpha$, where $\alpha$ is a constant less than 1 then $\tau_{mix}\leq \lceil \log_{1/\alpha}{2e} \rceil$.
 \end{proposition}

 \begin{proposition}\label{fact2}\cite{aldous1995reversible}
 If at least $\beta N$ entries in column of $P$ are bounded below by $\gamma/N$, then $d(P) \leq 1- \gamma[1- 2(1- \beta)]$, where $\beta > \frac{1}{2}$ and $\gamma > 0$.
 \end{proposition}

\noindent {\bf {\em Random walk on a periodic lattice.---}} 
We will start with discussing the standard random walk on $\mathcal{L}$. The mixing time proof of a random walk on the periodic lattices $\mathbb{Z}_n^d$ are based on the coupling idea. The coupling is defined as given a Markov chain $P$ on state space $S$, a Markovian coupling of two $P$-chains is a Markov chain $\{ (X_t, Y_t)\}_{t\geq 0}$ with state space $ S \times S $ satisfies,
\begin{align*}
\textbf{P}\{X_{t+1} &= a' | X_{t} = a, Y_{t = b}\} = P(a',a) \\
   \textbf{P}\{Y_{t+1} &= b' | X_{t} = a, Y_{t = b}\} = P(b',b),    
\end{align*}
for all $ a, a', b, b'$. It has been proved that \textit{On a d-dimensional torus $\mathbb{Z}_{n}^d$ the upper bound on $\epsilon$- mixing time for $\epsilon < \frac{1}{2}$ of a lazy random walk is $dn^2\lceil \log_{4}(\frac{d}{\epsilon})\rceil$}~\cite{levin2017markov}. In lazy random walk, we do a standard random walk with $1/2$ probability and stay at the current vertex with $1/2$ probability. The classical mixing time of a random walk on $\mathcal{L}$ is proved using the same line of argument as done for $\mathbb{Z}_n^d$ in Ref~(\cite{levin2017markov}).

In $\mathcal{L} = \mathbb{Z}_{n_1} \times \mathbb{Z}_{n_2} \times \dots \times \mathbb{Z}_{n_d}$ where $n_1, n_2, \dots, n_d \in \mathbb{N}$, vertex $x = (x_1, \dots , x_d)$ and $y = (y_1, \dots, y_d)$ are neighbours if for some $j \in \{1, 2, \dots, d \}$ we have $x_i = y_i$ for $i \neq j$ and $x_j = y_j \pm 1 \text{ mod } n_j$. To avoid the same problem of periodicity as before, we do a lazy
random walk on $\mathcal{L}$. Without loss of generality, we assume that $n_1 > n_l$ were $2 \leq l \leq d$.

\begin{theorem}\label{thm1}
On $\mathbb{Z}_{n_1} \times \mathbb{Z}_{n_2} \times \dots \times \mathbb{Z}_{n_d} $ where $n_1, n_2, \dots, n_d \in \mathbb{N}$ and $n_1 > n_l$ for $2 \leq l \leq d$, the upper bound on $\epsilon$- mixing time for $\epsilon < \frac{1}{2}$ of a lazy random walk is $2dn_{1}^2\lceil \log(\frac{d}{\epsilon})\rceil$. 
\end{theorem}

The proof is given in appendix \ref{classical}. Here notice that classical random walk takes time $O(d n_1^2 \log{(d/\epsilon)})$ to mix on $\mathcal{L}$. This result motivated the study of quantum walk on general periodic lattice $\mathcal{L}$.

\noindent {\bf {\em Coordinate-wise quantum walks.---}}
 We observed that if we do coordinate-wise mixing using continuous time quantum walk on $\mathcal{L}$, then it is quadratically better than the classical mixing time in Theorem~\ref{thm1}. In this continuous time quantum walk, we mix on each $\mathbb{Z}_{n_i}$ copy separately. We show that doing the quantum walk this way mixes in time $O\Big(\Big(\sum_{i=1}^{d} n_i \Big) (\log{d/\epsilon})\Big)$  for $\mathcal{L}$. We use a lemma from~\cite{PR1} about mixing time on $n-cycle$ to prove this result. If a graph is regular, the simple random walk $P$ is the graph's normalized adjacency matrix $\Bar{A}$.  

\begin{lemma}\label{lemma:lemma1}\cite{PR1}
Let $\mathbb{Z}_{n}$ be the cycle on $n \geq 2$ vertices. The continuous time quantum walks $U = e^{iPt}$, where $P$ is the simple random walk matrix on $\mathbb{Z}_{n}$ with instantaneous or repeated  measurements, mixes for any time $ t \in I :=[\frac{n}{3}, \frac{n}{2}]$.
\end{lemma}

\begin{theorem}\label{thm2}
Continuous time quantum walk mixes on $\mathbb{Z}_{n_1} \times \mathbb{Z}_{n_2} \times \dots \times \mathbb{Z}_{n_d}$ in time of order $O\Big(\Big(\sum_{i=1}^{d} n_i \Big) \log{(d/\epsilon)}\Big)$.
\end{theorem}

\begin{proof}

 The unitary operators to evolve the quantum walk coordinate-wise are $U_k = e^{i H_k t}$ where $H_k = I_{n_1} \otimes \dots \otimes \Bar{A}_k \otimes \dots \otimes I_{n_d} $ for $1 \leq k \leq d$. Here the initial state is $\rho_0 = \rho^{0}_{1} \otimes \rho^{0}_{2} \otimes \dots \otimes \rho^{0}_{d}$ where $\rho^0_{l} = \ket{0}\bra{0}_l$ for $1 \leq l \leq d$. 

Now applying $U_1$ on $\rho_0$ gives us

 \begin{equation}
        \begin{split}
            \rho_t & = e^{i H_1 t} \rho_0 e^{-i H_1 t}\\
             & = e^{i A_1 t} \rho^0_{1} e^{-i A_1 t} \otimes I_{n_2}\rho^0_{2}I_{n_2} \otimes I_{n_3}\rho^0_{3}I_{n_3} \otimes \dots \otimes I_{n_2}\rho^0_{d}I_{n_d}\\
             & = e^{i A_1 t} \rho^0_{1} e^{-i A_1 t} \otimes \rho^0_{2} \otimes \rho^0_{3} \otimes \dots \otimes\rho^0_{d}.
        \end{split}
    \end{equation}

When we run the quantum walk for time $t = O(n_1)$, in the first coordinate we get  $|\bra{j_1}e^{i \Bar{A}_1 \frac{n_1}{3}} \ket{0}_1| = \Theta \left(\frac{1}{\sqrt{n_1}}\right)$ for at least $\frac{2}{3}$ position $\ket{j_1}$, where $j_1 \in [0,n_1-1]$ (Ref.~\cite{PR1} and lemma~\ref{lemma:lemma1} on $n$-cycle). After the measurement in the position basis of $\mathbb{Z}_{n_1}$ we have $\frac{2}{3}$ diagonal entries from $\rho^{t}_{1}$ are $\Omega\left(\frac{1}{n_1}\right)$. Consecutively, as shown in Figure \ref{evolution}, run the quantum walk using $U_2$, $U_3$,\dots, $U_d$ for time $\frac{n_2}{3} $, $\frac{n_3}{3} $, \dots, $\frac{n_d}{3}$ and measure in their position basis respectively. 

\begin{figure}[ht]
    \centering
    \begin{tikzpicture}
\draw[gray, thick, ->] (-7.5,0) -- (0.9,0);
\filldraw [black] (-5.7,0) circle (2pt)node[anchor=south]{$ t =O(n_1)$};
\node[draw] at (-7,0.5) {$U_1$};
\node[draw] at (-4.4,0.5) {$U_2$};
\node[draw] at (-5.7,-0.5) {measure};
\filldraw [black] (-3.1,0) circle(2pt)node[anchor=south]{$t =O(n_2)$};
\node[draw] at (-3.1,-0.5) {measure};
\filldraw [black] (-7.5,0) circle (2pt)node[anchor=north]{$\rho_0$};
\node[draw] at (-1.9,0.5) {$\dots$};
\node[draw] at (-1.1,0.5) {$U_{d}$};
\filldraw [black] (0.1,0) circle (2pt)node[anchor=south]{$t =O(n_d)$};
\node[draw] at (0.1,-0.5) {measure};
\end{tikzpicture}
    \caption{\textbf{The timeline of quantum evolution:} Here, the black points represent a measurement at time $t$ and the line connecting the points shows unitary evolution.}
    \label{evolution}
\end{figure}
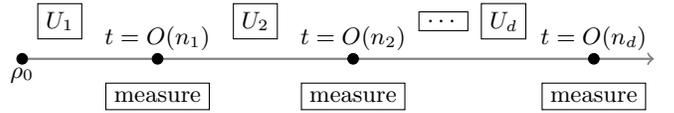
\vspace{-\baselineskip} Using propositions \ref{fact1} and \ref{fact2}, we achieve threshold mixing on each coordinate by repeating the above process a constant number of times. The state of this computation at any stage is a product state across coordinates, i.e., $\rho_1 \otimes \rho_2 \otimes \dots \otimes \rho_d$. As in Richter's analysis~\cite{PR1}, by repeating the above computation $\log(d/\epsilon)$ times, the $\ell$-th component is within $\epsilon/d$ of uniform distribution on $\mathbb{Z}_{n_l}$ for $1 \leq \ell \leq d$. This means that, in the end, the probability distribution $P'_\ell$ produced on $\mathbb{Z}_{n_l}$ in the $\ell$-th component satisfies
$\frac{1}{2}\parallel P'_{\ell} - u_\ell  \parallel_1 \leq \epsilon/d$, where $u_\ell$ is the uniform distribution with each entry on the diagonal $\frac{1}{n_l}$. Thus, the final state after time $T = O(\sum_j n_j \log{(d/\epsilon)})$ has the form $\rho_T = \rho^T_1 \otimes \rho^T_2 \otimes \dots \otimes \rho^T_d$, where each $\rho^T_\ell$ is a diagonal state wrt the standard basis, and within $2\epsilon/d$ of $I_{n_\ell}/n_\ell$. It follows that the distribution on the states given by $\rho_T$ is within $\epsilon$ of the maximally mixed state on $\mathbb{Z}_{n_1} \times \mathbb{Z}_{n_2} \times \dots \times \mathbb{Z}_{n_d}$.
\end{proof}

We consider a particular case when our result takes only extra $\log{d}$ time than~\cite{PR1} when $\mathbb{Z}_{n_1} \times \mathbb{Z}_{n_2} \times \dots \times \mathbb{Z}_{n_d}$ has $n_i$ for $i\geq 2$ of order $n_1$. Let us consider $n_i = \frac{n_1}{i}$ for all $i\geq 2$. Then using Theorem \ref{thm1}, we get the mixing time 

\begin{equation}
\begin{split}
   \Big(\frac{2}{3}\Big(\sum_{i=1}^{d} n_i \Big) \log{(d/\epsilon)}\Big) & = \Big(\frac{2}{3}\Big(\sum_{i=1}^{d} \frac{n_1}{i} \Big) \log{(d/\epsilon)}\Big)\\
   & = \Big(\frac{2}{3}\Big(n_1 \sum_{i=1}^{d} \frac{1}{i} \Big) \log{(d/\epsilon)}\Big)\\
   & = \Big(\frac{2}{3}\Big(n_1 H_d\Big) \log{(d/\epsilon)}\Big)\\
   & \leq \Big(\frac{2}{3} n_1(\log{(d)}+1)\log{(d/\epsilon)}\Big),
\end{split}
\end{equation}

where $H_d = \sum_{i=1}^{d} \frac{1}{i}$ is a Harmonic number function  is upper bounded by $\log{(d)}+1$. The mixing time in this special case is of order $O\left(n_1 \log{(d)} \log{(d/\epsilon)}\right)$ whereas in in Richter's case it is $O(nd \log{d}\log{(1/\epsilon)})$.

\noindent {\bf {\em Quantum walks with repeated measurements.---}}\label{repeatedq} 
In this section, we analyse the walk on $\mathbb{Z}_{n_1} \times \mathbb{Z}_{n_2}$ and subsequently state the conjecture on the trigonometric sum. We then give a proof for mixing time on $\mathbb{Z}_{n_1} \times \mathbb{Z}_{n_2}$ and end this section with a hypothesis for the general case $\mathcal{L}$. We describe the continuous time quantum walk algorithm for $\mathcal{L}$ in the appendix~\ref{qalgo}. We provide a discussion about quantum walks on $\mathbb{Z}_n$ in appendix~\ref{zn} that helps to do the direct analysis on $\mathcal{L}$.       

Let us assume $n_1$ and $n_2$ are \emph{relatively prime and odd}. The simple classical random walk on $\mathbb{Z}_{n_1} \times \mathbb{Z}_{n_2}$ is given by the transition matrix
\begin{equation}
    \Bar{A}' = \frac{1}{2}(\Bar{A_1} \otimes I_{2} +I_{1} \otimes \Bar{A_2} ),
\end{equation}
where $\Bar{A_1}$ is the adjacency matrix of the graph on $\mathbb{Z}_{n_1}$
and $\Bar{A_2}$ is the adjacency matrix of the graph on $\mathbb{Z}_{n_2}$
(see appendix~\ref{zn}, Eq.~(\ref{Eq;Adj})).
The quantum walk operator for time $t$ is 
\begin{align}
  e^{i\Bar{A'}t} & = e^{\frac{it}{2}(\Bar{A_1} \otimes I_{2} +I_{1} \otimes \Bar{A_2} )},\\
  & = e^{\frac{it}{2}\Bar{A_1}} \otimes e^{\frac{it}{2}\Bar{A_2}}.
\end{align}

Suppose we start the walk $\ket{0} = \ket{0_1}\ket{0_2}$, evolve it for a time $t$ chosen uniformly from $[0,T]$. We are interested in the probability that the final state is $\ket{l} = \ket{l_1}\ket{l_2}$, that is,

\begin{equation}\label{probl}
 \begin{aligned}
    P_T(0, l) & = \frac{1}{T} \int_{0}^{T} |\bra{l}e^{i\Bar{A'}t}\ket{0}|^2 dt,\\
    & = \frac{1}{T} \int_{0}^{T} |\bra{l_1}e^{i\Bar{A}_1t/2}\ket{0_1}|^2|\bra{l_2}e^{i\Bar{A}_2t/2}\ket{0_2}|^2 dt,\\
    & = \frac{1}{T} \int_{0}^{T} \frac{1}{(n_1n_2)^2}\Bigg|\sum_{j_1 = 0}^{n_1 -1} e^{\frac{it}{2} \cos(2 \pi j_1/n_1)} w_{1}^{l_1j_1}\Bigg|^2\\\
   & \times \Bigg|\sum_{j_2 = 0}^{n_2 -1} e^{\frac{it}{2} \cos(2 \pi j_2/n_2)} w_{2}^{l_2j_2}\Bigg|^2 dt,
\end{aligned}   
\end{equation}
where $w_1 = e^{2\pi i/n_1}$ and $w_2 = e^{2\pi i/n_2}$. Using formulation in appendix~(\ref{zn}), and definition of $n_1(t)$, $n_2(t)$, $n_1 {\cal{I}}\{l_1=0\}-1$, and $n_2 {\cal{I}}\{l_1=0\}-1$  we write $P_T(0, l)$ as follows.

\begin{equation}\label{probl2}
    \begin{aligned}
& = \frac{1}{T} \int_{0}^{T} \frac{1}{(n_1n_2)^2} \Big[
\Big( n_1 + n_1 {\cal{I}}\{l_1=0\}-1 + n_1(t)\\
& \times \Big( n_2 + n_2 {\cal{I}}\{l_2=0\}-1 + n_2(t)\Big) \Big] dt .
\end{aligned}
\end{equation}

To prove the mixing time bound on $\mathbb{Z}_{n_1} \times \mathbb{Z}_{n_2}$ in section~\ref{repeatedq}, we need the following lemma.

\begin{lemma}\label{lemma2}
  \begin{equation}\label{case:n1}
  \left|\int_{0}^{T} n_1(t) dt \right|\leq 32(n_1 \log(n_1))^2,
\end{equation}

\begin{equation}\label{case:n2}
   \left|\int_{0}^{T} n_2(t) dt\right| \leq  32 (n_2 \log(n_2))^2.
\end{equation}  
\end{lemma}
\begin{proof}
 Please refer appendix~\ref{lemma2app}.    
\end{proof} 
This lemma gives the bound on the integration of single time-dependent terms. The bound on integral of multiplication of such terms is conjectured in the subsequent section.  

\noindent {\bf {\em Conjecture on trigonometric sum.---}}
To prove quantum walk mixing on $\mathcal{L}$, we present a conjecture on the value of a trigonometric sum that arises in our analysis of the quantum mixing time for the general case $d$. We consider the following equation
\begin{equation}\label{conj}
   \left|\int_{0}^{T} n_1(t) n_2(t) \dots n_d(t) dt\right|,
\end{equation}

To prove an upper bound on the mixing time of the quantum walk on $\mathcal{L}$, we propose the following conjecture for the general case: 

\begin{conjecture}\label{conj1}
For $n_1 > n_2 > \dots > n_d$, where $n_1, n_2, \dots, n_d$ are odd and relatively prime, we have

\begin{equation}
\begin{aligned}
    &\left|\int_0^T \prod_{i=1}^d n_i(t)dt \right|\leq 16 d \sum_{j= 1}^{d} \left(\prod_{i \neq j}^d n_i(t)\right)  (n_j\log(n_j))^2.
\end{aligned}    
\end{equation}
\end{conjecture}

In the proof of Theorem \ref{thm3} we use specific case $d=2$ of the conjecture, i.e., for $\mathbb{Z}_{n_1}\times \mathbb{Z}_{n_2}$:

\begin{equation}\label{d2}
  \left| \int_{0}^{T} n_1(t) n_2(t) dt \right|\leq 32 n_1(n_2 \log(n_2))^2+ 32 n_2(n_1 \log(n_1))^2. 
\end{equation}

To support our conjecture, we carried out numerical simulations. The details are explained in appendix~\ref{simul}.   

\begin{theorem}\label{thm3}
For time T of order $O((n_1+n_2)(\log{(n_1)})^2)$,$T' = O\left(\log{(1/\epsilon)}\right)$ repeated continuous time quantum walk on the graph $\mathbb{Z}_{n_1} \times \mathbb{Z}_{n_2}$ ($n_1 > n_2$, $n_2 > 91$, $n_1, n_2$ are odd and relatively prime) mixes to the uniform distribution.
\end{theorem} 

\begin{proof}

The quantum walk algorithm discussed in appendix~\ref{qalgo} suggests that we want to minimize $T T'$ such that, 

\begin{equation}\label{bdp2}
\frac{1}{2}\parallel (P_{T})^{T'} - \mathcal{u} \parallel_{1} \leq d((P_{T})^{T'}) \leq \epsilon,
\end{equation}

where $\mathcal{u} = \frac{1}{n_1 n_2}\mathbf{J}$ and $\mathbf{J}$ is all ones matrix. From Ref.~\cite{Richter_2007} we have 
\begin{equation}\label{bdp}
    \frac{1}{2}\parallel P_{T} - \mathcal{u}\parallel_{1} = \max_{j}\parallel P_T(:,j) - \mathcal{u}(:,j) \parallel_{1}.  
\end{equation}
Note that since the Adjacency matrix $\Bar{A'}$ is symmetric so $P_T(x,y) = P_T(y,x)$. Hence, each column of $P_T$ has the same probability distribution up to permutation, and all columns of $\mathcal{u}$ are the same, respectively. Then \emph{w.l.o.g.} we choose the first column of $P_T$ and $\mathcal{u}$ and rewrite Eq.~\ref{bdp} as

\begin{equation}
 \frac{1}{2}\parallel P_{T} - \mathcal{u}\parallel_{1} = \parallel P_T(:,1) - \mathcal{u}(:,1) \parallel_{1} = d(P_T).   
\end{equation} 
Using Eq.~\ref{probl} and Eq.~\ref{probl2} we write $\parallel P_T(:,1) - \mathcal{u}(:,1) \parallel_{1}$ as follows: 
\begin{equation}
 \begin{aligned}
    &\parallel P_{T}(:,1) - \mathcal{u}(:,1)\parallel_{1}  \\
    &=\sum_{l}\Bigg|\frac{1}{T} \int_{0}^{T} \frac{1}{(n_1n_2)^2} \Bigg( n_1 (n_2 {\cal{I}}\{l_2=0\}-1) + n_1 n_2(t)\\
    &+ n_2 (n_1 {\cal{I}}\{l_1=0\}-1)+ (n_1 {\cal{I}}\{l_1=0\}-1) \\
    & \times(n_2 {\cal{I}}\{l_2=0\}-1)+ (n_1{\cal{I}}\{l_1=0\}-1) n_2(t) \\
    &+ n_2 n_1(t) +(n_2{\cal{I}}\{l_2=0\}-1) n_1(t) + n_1(t) n_2(t)  \Bigg)dt \Bigg|.
\end{aligned}   
\end{equation}

Now using lemma~\ref{lemma2} and Eq.~\ref{d2} we prove the bound in Eq.~\ref{bdp2}. Note that column $P_T(:,1)$ means starting vertex is $0$ and $l = l_1 l_2$ denotes row indices where $0 \leq l \leq n_1n_2-1$. We make cases based on values of $l_1$ and $l_2$ and take time  $T = 1600(n_1+n_2)(\log{n_1})^2$ for all cases. The detailed analysis of the following bounds is given in appendix~\ref{bounds}. 

\textbf{ case 0: $l_1 = 0$, $l_2 = 0$ } 

\begin{equation}
\left|P_{T}(0,0) - \frac{1}{n_1n_2}\right| \leq \frac{4}{n_2^2}.     
\end{equation}

\textbf{case 1 : $l_1 = 0$ and $l_2 \neq 0$} (There are $(n_2 -1)$ such terms.)

\begin{equation}
 n_2\left|P_T(0,l) - \frac{1}{n_1n_2}\right| \leq \frac{3}{n_2}.    
\end{equation}

\textbf{case 2 : $l_1 \neq 0$ and $l_2 = 0$} (There are $(n_1 -1)$ such terms.)

\begin{equation}
  n_1\left|P_T(0,l) - \frac{1}{n_1n_2}\right| \leq \frac{3}{n_2}.  
\end{equation}

\textbf{case 3 : $l_1 \neq 0$ and $l_2 \neq 0$} (There are $(n_1 n_2 -n_1- n_2 + 1)$ such terms.) 

\begin{equation}
  n_1 n_2 \left|P_T(0,l) - \frac{1}{n_1n_2}\right| \leq \frac{3}{n_2}+ \frac{2}{50}.   
\end{equation}

The bounds combined from all the above cases give us the following.  
\begin{equation}
\begin{aligned}
    \parallel P_T(:,1) - \mathcal{u}(:,1)\parallel_1 & \leq \frac{4}{n_2^2}+ \frac{3}{n_2} + \frac{3}{n_2} + \frac{3}{n_2}+ \frac{2}{50},\\
    & \leq \frac{13}{n_2} + \frac{2}{50}.
\end{aligned}    
\end{equation}

This implies 

\begin{equation}
    d(P_T) \leq \frac{1}{2e},
\end{equation}

for $n_2 > 91$. Now we  want $d(P_{T})^{T'} \leq \epsilon$. Then by submultiplicative property of Markov chain we know  $d((P_{T})^{T'}) \leq d(P_{T})^{T'}$ and it gives us

\begin{equation}
\begin{aligned}
    \implies & \left(\frac{1}{2e}\right)^{T'} < \epsilon\\
    & T' \geq \log_{2e}{(1/\epsilon)}.
\end{aligned}   
\end{equation}
   
\end{proof}

In the end, based on the work for $d = 2$ case, we conjecture that for that any $d > 2$, the continuous time quantum walk takes $O(\sum_{i = 1}^{d} n_i (\log{(n_1)})^2\log{(1/\epsilon)})$  time to mix on $\mathcal{L}$ with $O(\log{(1/\epsilon)})$ measurements.

\begin{conjecture}\label{general}
 For time of order $O(\sum_{i = 1}^{d} n_i (\log{(n_1)})^2\log{(1/\epsilon)})$ repeated continuous time quantum walk on the graph $\mathbb{Z}_{n_1} \times \mathbb{Z}_{n_2} \times \dots \times \mathbb{Z}_{n_d}$ ($n_1 > n_2 > \dots > n_d$ odd and relatively prime) mixes to the uniform distribution.   
\end{conjecture}

\noindent {\bf {\em Summary and outlook.---}}  
We summarize this study with some important highlights. This work focuses on quantum mixing time and presents two types of quantum walks that differ by the number of measurements required. We provide a version of quantum walk with coordinate-wise mixing and show that it takes $O(\sum_{i = 1}^d n_i \log{(d/\epsilon)})$ time and $O(d \log{(d/\epsilon)})$ measurements. Importantly, we present a direct analysis of the standard quantum walk on $\mathcal{L}$, which achieves a mixing time close to the coordinate-wise mixing time result with only $O(\log{(1/\epsilon)})$ measurements.

 To conclude, we propose a conjecture for the general case of $\mathcal{L}$ and estimate that the expected mixing time is $O(\sum_{i = 1}^{d} n_i (\log{(n_1)})^2\log{(1/\epsilon)})$ with $O(\log{(1/\epsilon)})$ measurements. This conjecture is a step towards proving the general conjecture on quantum walk mixing time on regular graphs. Our work extends the class of classical Markov chain Monte Carlo processes where quantum walks with repeated measurements have a speedup advantage and its potential applications in numerical approximations, computational physics and computational biology. 

 \noindent {\bf {\em Acknowledgements.---}} We thank Professor Jaikumar Radhakrishnan for his guidance and support throughout this project and for hosting Shyam Dhamapurkar at the Tata Institute of Fundamental Research, Mumbai. We gratefully acknowledge discussions with Ashwin Nayak, Oscar Dahlsten, and Saniya Wagh.
This work was supported by the Key-Area Research and Development Program of Guang-Dong Province (Grant No. 2018B030326001), Shenzhen Science and Technology Program (KQTD20200820113010023).

\nocite{*}
\bibliographystyle{apsrev4-1}
\bibliography{apssamp}

\begin{thebibliography}{31}%
\makeatletter
\providecommand \@ifxundefined [1]{%
 \@ifx{#1\undefined}
}%
\providecommand \@ifnum [1]{%
 \ifnum #1\expandafter \@firstoftwo
 \else \expandafter \@secondoftwo
 \fi
}%
\providecommand \@ifx [1]{%
 \ifx #1\expandafter \@firstoftwo
 \else \expandafter \@secondoftwo
 \fi
}%
\providecommand \natexlab [1]{#1}%
\providecommand \enquote  [1]{``#1''}%
\providecommand \bibnamefont  [1]{#1}%
\providecommand \bibfnamefont [1]{#1}%
\providecommand \citenamefont [1]{#1}%
\providecommand \href@noop [0]{\@secondoftwo}%
\providecommand \href [0]{\begingroup \@sanitize@url \@href}%
\providecommand \@href[1]{\@@startlink{#1}\@@href}%
\providecommand \@@href[1]{\endgroup#1\@@endlink}%
\providecommand \@sanitize@url [0]{\catcode `\\12\catcode `\$12\catcode
  `\&12\catcode `\#12\catcode `\^12\catcode `\_12\catcode `\%12\relax}%
\providecommand \@@startlink[1]{}%
\providecommand \@@endlink[0]{}%
\providecommand \url  [0]{\begingroup\@sanitize@url \@url }%
\providecommand \@url [1]{\endgroup\@href {#1}{\urlprefix }}%
\providecommand \urlprefix  [0]{URL }%
\providecommand \Eprint [0]{\href }%
\providecommand \doibase [0]{http://dx.doi.org/}%
\providecommand \selectlanguage [0]{\@gobble}%
\providecommand \bibinfo  [0]{\@secondoftwo}%
\providecommand \bibfield  [0]{\@secondoftwo}%
\providecommand \translation [1]{[#1]}%
\providecommand \BibitemOpen [0]{}%
\providecommand \bibitemStop [0]{}%
\providecommand \bibitemNoStop [0]{.\EOS\space}%
\providecommand \EOS [0]{\spacefactor3000\relax}%
\providecommand \BibitemShut  [1]{\csname bibitem#1\endcsname}%
\let\auto@bib@innerbib\@empty
\bibitem [{\citenamefont {Venegas-Andraca}(2012)}]{venegas2012quantum}%
  \BibitemOpen
  \bibfield  {author} {\bibinfo {author} {\bibfnamefont {S.~E.}\ \bibnamefont
  {Venegas-Andraca}},\ }\href@noop {} {\bibfield  {journal} {\bibinfo
  {journal} {Quantum Information Processing}\ }\textbf {\bibinfo {volume}
  {11}},\ \bibinfo {pages} {1015} (\bibinfo {year} {2012})}\BibitemShut
  {NoStop}%
\bibitem [{\citenamefont {Kadian}\ \emph {et~al.}(2021)\citenamefont {Kadian},
  \citenamefont {Garhwal},\ and\ \citenamefont {Kumar}}]{kadian2021quantum}%
  \BibitemOpen
  \bibfield  {author} {\bibinfo {author} {\bibfnamefont {K.}~\bibnamefont
  {Kadian}}, \bibinfo {author} {\bibfnamefont {S.}~\bibnamefont {Garhwal}}, \
  and\ \bibinfo {author} {\bibfnamefont {A.}~\bibnamefont {Kumar}},\
  }\href@noop {} {\bibfield  {journal} {\bibinfo  {journal} {Computer Science
  Review}\ }\textbf {\bibinfo {volume} {41}},\ \bibinfo {pages} {100419}
  (\bibinfo {year} {2021})}\BibitemShut {NoStop}%
\bibitem [{\citenamefont {Ambainis}(2007)}]{Ambainis_element_distinctness}%
  \BibitemOpen
  \bibfield  {author} {\bibinfo {author} {\bibfnamefont {A.}~\bibnamefont
  {Ambainis}},\ }\href@noop {} {\bibfield  {journal} {\bibinfo  {journal} {SIAM
  Journal on Computing}\ }\textbf {\bibinfo {volume} {37}},\ \bibinfo {pages}
  {210} (\bibinfo {year} {2007})}\BibitemShut {NoStop}%
\bibitem [{\citenamefont {Magniez}\ \emph {et~al.}(2005)\citenamefont
  {Magniez}, \citenamefont {Santha},\ and\ \citenamefont
  {Szegedy}}]{Magniez05quantumalgorithms}%
  \BibitemOpen
  \bibfield  {author} {\bibinfo {author} {\bibfnamefont {F.}~\bibnamefont
  {Magniez}}, \bibinfo {author} {\bibfnamefont {M.}~\bibnamefont {Santha}}, \
  and\ \bibinfo {author} {\bibfnamefont {M.}~\bibnamefont {Szegedy}},\ }in\
  \href@noop {} {\emph {\bibinfo {booktitle} {PROCEEDINGS OF SODA’05}}}\
  (\bibinfo {year} {2005})\ pp.\ \bibinfo {pages} {1109--1117}\BibitemShut
  {NoStop}%
\bibitem [{\citenamefont {Childs}\ and\ \citenamefont
  {Eisenberg}(2003)}]{Childs_Subsetfinding}%
  \BibitemOpen
  \bibfield  {author} {\bibinfo {author} {\bibfnamefont {A.~M.}\ \bibnamefont
  {Childs}}\ and\ \bibinfo {author} {\bibfnamefont {J.~M.}\ \bibnamefont
  {Eisenberg}},\ }\href@noop {} {\bibfield  {journal} {\bibinfo  {journal}
  {arXiv preprint quant-ph/0311038}\ } (\bibinfo {year} {2003})}\BibitemShut
  {NoStop}%
\bibitem [{\citenamefont {Farhi}\ and\ \citenamefont
  {Gutmann}(1998{\natexlab{a}})}]{Farhi_1998}%
  \BibitemOpen
  \bibfield  {author} {\bibinfo {author} {\bibfnamefont {E.}~\bibnamefont
  {Farhi}}\ and\ \bibinfo {author} {\bibfnamefont {S.}~\bibnamefont
  {Gutmann}},\ }\href@noop {} {\bibfield  {journal} {\bibinfo  {journal}
  {Physical Review A}\ }\textbf {\bibinfo {volume} {58}},\ \bibinfo {pages}
  {915–928} (\bibinfo {year} {1998}{\natexlab{a}})}\BibitemShut {NoStop}%
\bibitem [{\citenamefont {Yan}\ \emph {et~al.}(2019)\citenamefont {Yan},
  \citenamefont {Zhang}, \citenamefont {Gong}, \citenamefont {Wu},
  \citenamefont {Zheng}, \citenamefont {Li}, \citenamefont {Wang},
  \citenamefont {Liang}, \citenamefont {Lin}, \citenamefont {Xu} \emph
  {et~al.}}]{yan2019strongly}%
  \BibitemOpen
  \bibfield  {author} {\bibinfo {author} {\bibfnamefont {Z.}~\bibnamefont
  {Yan}}, \bibinfo {author} {\bibfnamefont {Y.-R.}\ \bibnamefont {Zhang}},
  \bibinfo {author} {\bibfnamefont {M.}~\bibnamefont {Gong}}, \bibinfo {author}
  {\bibfnamefont {Y.}~\bibnamefont {Wu}}, \bibinfo {author} {\bibfnamefont
  {Y.}~\bibnamefont {Zheng}}, \bibinfo {author} {\bibfnamefont
  {S.}~\bibnamefont {Li}}, \bibinfo {author} {\bibfnamefont {C.}~\bibnamefont
  {Wang}}, \bibinfo {author} {\bibfnamefont {F.}~\bibnamefont {Liang}},
  \bibinfo {author} {\bibfnamefont {J.}~\bibnamefont {Lin}}, \bibinfo {author}
  {\bibfnamefont {Y.}~\bibnamefont {Xu}},  \emph {et~al.},\ }\href@noop {}
  {\bibfield  {journal} {\bibinfo  {journal} {Science}\ }\textbf {\bibinfo
  {volume} {364}},\ \bibinfo {pages} {753} (\bibinfo {year}
  {2019})}\BibitemShut {NoStop}%
\bibitem [{\citenamefont {Wang}\ \emph {et~al.}(2022)\citenamefont {Wang},
  \citenamefont {Liu}, \citenamefont {Zhan}, \citenamefont {Xue}, \citenamefont
  {Zheng}, \citenamefont {Zeng}, \citenamefont {Wu}, \citenamefont {Wang},
  \citenamefont {Zheng}, \citenamefont {Wang} \emph {et~al.}}]{wang2022large}%
  \BibitemOpen
  \bibfield  {author} {\bibinfo {author} {\bibfnamefont {Y.}~\bibnamefont
  {Wang}}, \bibinfo {author} {\bibfnamefont {Y.}~\bibnamefont {Liu}}, \bibinfo
  {author} {\bibfnamefont {J.}~\bibnamefont {Zhan}}, \bibinfo {author}
  {\bibfnamefont {S.}~\bibnamefont {Xue}}, \bibinfo {author} {\bibfnamefont
  {Y.}~\bibnamefont {Zheng}}, \bibinfo {author} {\bibfnamefont
  {R.}~\bibnamefont {Zeng}}, \bibinfo {author} {\bibfnamefont {Z.}~\bibnamefont
  {Wu}}, \bibinfo {author} {\bibfnamefont {Z.}~\bibnamefont {Wang}}, \bibinfo
  {author} {\bibfnamefont {Q.}~\bibnamefont {Zheng}}, \bibinfo {author}
  {\bibfnamefont {D.}~\bibnamefont {Wang}},  \emph {et~al.},\ }\href@noop {}
  {\bibfield  {journal} {\bibinfo  {journal} {arXiv preprint arXiv:2208.13186}\
  } (\bibinfo {year} {2022})}\BibitemShut {NoStop}%
\bibitem [{\citenamefont {Gong}\ \emph {et~al.}(2021)\citenamefont {Gong},
  \citenamefont {Wang}, \citenamefont {Zha}, \citenamefont {Chen},
  \citenamefont {Huang}, \citenamefont {Wu}, \citenamefont {Zhu}, \citenamefont
  {Zhao}, \citenamefont {Li}, \citenamefont {Guo} \emph
  {et~al.}}]{gong2021quantum}%
  \BibitemOpen
  \bibfield  {author} {\bibinfo {author} {\bibfnamefont {M.}~\bibnamefont
  {Gong}}, \bibinfo {author} {\bibfnamefont {S.}~\bibnamefont {Wang}}, \bibinfo
  {author} {\bibfnamefont {C.}~\bibnamefont {Zha}}, \bibinfo {author}
  {\bibfnamefont {M.-C.}\ \bibnamefont {Chen}}, \bibinfo {author}
  {\bibfnamefont {H.-L.}\ \bibnamefont {Huang}}, \bibinfo {author}
  {\bibfnamefont {Y.}~\bibnamefont {Wu}}, \bibinfo {author} {\bibfnamefont
  {Q.}~\bibnamefont {Zhu}}, \bibinfo {author} {\bibfnamefont {Y.}~\bibnamefont
  {Zhao}}, \bibinfo {author} {\bibfnamefont {S.}~\bibnamefont {Li}}, \bibinfo
  {author} {\bibfnamefont {S.}~\bibnamefont {Guo}},  \emph {et~al.},\
  }\href@noop {} {\bibfield  {journal} {\bibinfo  {journal} {Science}\ }\textbf
  {\bibinfo {volume} {372}},\ \bibinfo {pages} {948} (\bibinfo {year}
  {2021})}\BibitemShut {NoStop}%
\bibitem [{\citenamefont {Wang}\ \emph {et~al.}(2023)\citenamefont {Wang},
  \citenamefont {Jiang}, \citenamefont {Feng},\ and\ \citenamefont
  {Li}}]{wang2023universal}%
  \BibitemOpen
  \bibfield  {author} {\bibinfo {author} {\bibfnamefont {Q.}~\bibnamefont
  {Wang}}, \bibinfo {author} {\bibfnamefont {Y.}~\bibnamefont {Jiang}},
  \bibinfo {author} {\bibfnamefont {S.}~\bibnamefont {Feng}}, \ and\ \bibinfo
  {author} {\bibfnamefont {L.}~\bibnamefont {Li}},\ }\href@noop {} {\bibfield
  {journal} {\bibinfo  {journal} {arXiv preprint arXiv:2307.16133}\ } (\bibinfo
  {year} {2023})}\BibitemShut {NoStop}%
\bibitem [{\citenamefont {Xu}\ \emph {et~al.}(2022)\citenamefont {Xu},
  \citenamefont {Zhang},\ and\ \citenamefont {Li}}]{xu2022robust}%
  \BibitemOpen
  \bibfield  {author} {\bibinfo {author} {\bibfnamefont {Y.}~\bibnamefont
  {Xu}}, \bibinfo {author} {\bibfnamefont {D.}~\bibnamefont {Zhang}}, \ and\
  \bibinfo {author} {\bibfnamefont {L.}~\bibnamefont {Li}},\ }\href@noop {}
  {\bibfield  {journal} {\bibinfo  {journal} {Physical Review A}\ }\textbf
  {\bibinfo {volume} {106}},\ \bibinfo {pages} {052207} (\bibinfo {year}
  {2022})}\BibitemShut {NoStop}%
\bibitem [{\citenamefont {Levin}\ and\ \citenamefont
  {Peres}(2017)}]{levin2017markov}%
  \BibitemOpen
  \bibfield  {author} {\bibinfo {author} {\bibfnamefont {D.~A.}\ \bibnamefont
  {Levin}}\ and\ \bibinfo {author} {\bibfnamefont {Y.}~\bibnamefont {Peres}},\
  }\href@noop {} {\emph {\bibinfo {title} {Markov chains and mixing times}}},\
  Vol.\ \bibinfo {volume} {107}\ (\bibinfo  {publisher} {American Mathematical
  Soc.},\ \bibinfo {year} {2017})\BibitemShut {NoStop}%
\bibitem [{\citenamefont {Richter}(2007{\natexlab{a}})}]{Richter_2007}%
  \BibitemOpen
  \bibfield  {author} {\bibinfo {author} {\bibfnamefont {P.~C.}\ \bibnamefont
  {Richter}},\ }\href@noop {} {\bibfield  {journal} {\bibinfo  {journal} {New
  Journal of Physics}\ }\textbf {\bibinfo {volume} {9}},\ \bibinfo {pages} {72}
  (\bibinfo {year} {2007}{\natexlab{a}})}\BibitemShut {NoStop}%
\bibitem [{\citenamefont {Richter}(2007{\natexlab{b}})}]{PR1}%
  \BibitemOpen
  \bibfield  {author} {\bibinfo {author} {\bibfnamefont {P.~C.}\ \bibnamefont
  {Richter}},\ }\href@noop {} {\bibfield  {journal} {\bibinfo  {journal}
  {Physical Review A}\ }\textbf {\bibinfo {volume} {76}},\ \bibinfo {pages}
  {042306} (\bibinfo {year} {2007}{\natexlab{b}})}\BibitemShut {NoStop}%
\bibitem [{\citenamefont {Childs}\ \emph
  {et~al.}(2003{\natexlab{a}})\citenamefont {Childs}, \citenamefont {Cleve},
  \citenamefont {Deotto}, \citenamefont {Farhi}, \citenamefont {Gutmann},\ and\
  \citenamefont {Spielman}}]{Childsexpo}%
  \BibitemOpen
  \bibfield  {author} {\bibinfo {author} {\bibfnamefont {A.~M.}\ \bibnamefont
  {Childs}}, \bibinfo {author} {\bibfnamefont {R.}~\bibnamefont {Cleve}},
  \bibinfo {author} {\bibfnamefont {E.}~\bibnamefont {Deotto}}, \bibinfo
  {author} {\bibfnamefont {E.}~\bibnamefont {Farhi}}, \bibinfo {author}
  {\bibfnamefont {S.}~\bibnamefont {Gutmann}}, \ and\ \bibinfo {author}
  {\bibfnamefont {D.~A.}\ \bibnamefont {Spielman}},\ }in\ \href@noop {} {\emph
  {\bibinfo {booktitle} {Proceedings of the thirty-fifth annual ACM symposium
  on Theory of computing}}}\ (\bibinfo {year} {2003})\ pp.\ \bibinfo {pages}
  {59--68}\BibitemShut {NoStop}%
\bibitem [{\citenamefont {Childs}\ \emph {et~al.}(2002)\citenamefont {Childs},
  \citenamefont {Farhi},\ and\ \citenamefont {Gutmann}}]{childs2002example}%
  \BibitemOpen
  \bibfield  {author} {\bibinfo {author} {\bibfnamefont {A.~M.}\ \bibnamefont
  {Childs}}, \bibinfo {author} {\bibfnamefont {E.}~\bibnamefont {Farhi}}, \
  and\ \bibinfo {author} {\bibfnamefont {S.}~\bibnamefont {Gutmann}},\
  }\href@noop {} {\bibfield  {journal} {\bibinfo  {journal} {Quantum
  Information Processing}\ }\textbf {\bibinfo {volume} {1}},\ \bibinfo {pages}
  {35} (\bibinfo {year} {2002})}\BibitemShut {NoStop}%
\bibitem [{\citenamefont {Aldous}\ and\ \citenamefont
  {Fill}(1995)}]{aldous1995reversible}%
  \BibitemOpen
  \bibfield  {author} {\bibinfo {author} {\bibfnamefont {D.}~\bibnamefont
  {Aldous}}\ and\ \bibinfo {author} {\bibfnamefont {J.}~\bibnamefont {Fill}},\
  }\href@noop {} {\enquote {\bibinfo {title} {Reversible markov chains and
  random walks on graphs},}\ } (\bibinfo {year} {1995})\BibitemShut {NoStop}%
\bibitem [{\citenamefont {Chung}\ and\ \citenamefont
  {Sternberg}(1993)}]{bucky1}%
  \BibitemOpen
  \bibfield  {author} {\bibinfo {author} {\bibfnamefont {F.}~\bibnamefont
  {Chung}}\ and\ \bibinfo {author} {\bibfnamefont {S.}~\bibnamefont
  {Sternberg}},\ }\href@noop {} {\bibfield  {journal} {\bibinfo  {journal}
  {American scientist}\ }\textbf {\bibinfo {volume} {81}},\ \bibinfo {pages}
  {56} (\bibinfo {year} {1993})}\BibitemShut {NoStop}%
\bibitem [{\citenamefont {Chakraborty}\ \emph {et~al.}(2020)\citenamefont
  {Chakraborty}, \citenamefont {Luh},\ and\ \citenamefont
  {Roland}}]{chakraborty2020fast}%
  \BibitemOpen
  \bibfield  {author} {\bibinfo {author} {\bibfnamefont {S.}~\bibnamefont
  {Chakraborty}}, \bibinfo {author} {\bibfnamefont {K.}~\bibnamefont {Luh}}, \
  and\ \bibinfo {author} {\bibfnamefont {J.}~\bibnamefont {Roland}},\
  }\href@noop {} {\bibfield  {journal} {\bibinfo  {journal} {Physical review
  letters}\ }\textbf {\bibinfo {volume} {124}},\ \bibinfo {pages} {050501}
  (\bibinfo {year} {2020})}\BibitemShut {NoStop}%
\bibitem [{\citenamefont {Gerhardt}\ and\ \citenamefont
  {Watrous}(2003)}]{gerhardt2003continuous}%
  \BibitemOpen
  \bibfield  {author} {\bibinfo {author} {\bibfnamefont {H.}~\bibnamefont
  {Gerhardt}}\ and\ \bibinfo {author} {\bibfnamefont {J.}~\bibnamefont
  {Watrous}}\ }(\bibinfo {organization} {Springer},\ \bibinfo {year} {2003})\
  pp.\ \bibinfo {pages} {290--301}\BibitemShut {NoStop}%
\bibitem [{\citenamefont {Childs}\ \emph
  {et~al.}(2003{\natexlab{b}})\citenamefont {Childs}, \citenamefont {Cleve},
  \citenamefont {Deotto}, \citenamefont {Farhi}, \citenamefont {Gutmann},\ and\
  \citenamefont {Spielman}}]{Childs_2003}%
  \BibitemOpen
  \bibfield  {author} {\bibinfo {author} {\bibfnamefont {A.~M.}\ \bibnamefont
  {Childs}}, \bibinfo {author} {\bibfnamefont {R.}~\bibnamefont {Cleve}},
  \bibinfo {author} {\bibfnamefont {E.}~\bibnamefont {Deotto}}, \bibinfo
  {author} {\bibfnamefont {E.}~\bibnamefont {Farhi}}, \bibinfo {author}
  {\bibfnamefont {S.}~\bibnamefont {Gutmann}}, \ and\ \bibinfo {author}
  {\bibfnamefont {D.~A.}\ \bibnamefont {Spielman}},\ }\href@noop {} {\bibfield
  {journal} {\bibinfo  {journal} {Proceedings of the thirty-fifth ACM symposium
  on Theory of computing - STOC ’03}\ } (\bibinfo {year}
  {2003}{\natexlab{b}})}\BibitemShut {NoStop}%
\bibitem [{\citenamefont {Childs}\ and\ \citenamefont
  {Goldstone}(2004)}]{Childs_2004}%
  \BibitemOpen
  \bibfield  {author} {\bibinfo {author} {\bibfnamefont {A.~M.}\ \bibnamefont
  {Childs}}\ and\ \bibinfo {author} {\bibfnamefont {J.}~\bibnamefont
  {Goldstone}},\ }\href@noop {} {\bibfield  {journal} {\bibinfo  {journal}
  {Physical Review A}\ }\textbf {\bibinfo {volume} {70}},\ \bibinfo {pages}
  {022314} (\bibinfo {year} {2004})}\BibitemShut {NoStop}%
\bibitem [{\citenamefont {Childs}(2009)}]{childs2009universal}%
  \BibitemOpen
  \bibfield  {author} {\bibinfo {author} {\bibfnamefont {A.~M.}\ \bibnamefont
  {Childs}},\ }\href@noop {} {\bibfield  {journal} {\bibinfo  {journal}
  {Physical review letters}\ }\textbf {\bibinfo {volume} {102}},\ \bibinfo
  {pages} {180501} (\bibinfo {year} {2009})}\BibitemShut {NoStop}%
\bibitem [{\citenamefont {Farhi}\ and\ \citenamefont
  {Gutmann}(1998{\natexlab{b}})}]{AnalogFarhi}%
  \BibitemOpen
  \bibfield  {author} {\bibinfo {author} {\bibfnamefont {E.}~\bibnamefont
  {Farhi}}\ and\ \bibinfo {author} {\bibfnamefont {S.}~\bibnamefont
  {Gutmann}},\ }\href@noop {} {\bibfield  {journal} {\bibinfo  {journal}
  {Physical Review A}\ }\textbf {\bibinfo {volume} {57}},\ \bibinfo {pages}
  {2403} (\bibinfo {year} {1998}{\natexlab{b}})}\BibitemShut {NoStop}%
\bibitem [{\citenamefont {Kendon}\ and\ \citenamefont
  {Tregenna}(2003)}]{kendon2003decoherence}%
  \BibitemOpen
  \bibfield  {author} {\bibinfo {author} {\bibfnamefont {V.}~\bibnamefont
  {Kendon}}\ and\ \bibinfo {author} {\bibfnamefont {B.}~\bibnamefont
  {Tregenna}},\ }\href@noop {} {\bibfield  {journal} {\bibinfo  {journal}
  {Physical Review A}\ }\textbf {\bibinfo {volume} {67}},\ \bibinfo {pages}
  {042315} (\bibinfo {year} {2003})}\BibitemShut {NoStop}%
\bibitem [{\citenamefont {Kendon}(2007)}]{kendon2007decoherence}%
  \BibitemOpen
  \bibfield  {author} {\bibinfo {author} {\bibfnamefont {V.}~\bibnamefont
  {Kendon}},\ }\href@noop {} {\bibfield  {journal} {\bibinfo  {journal}
  {Mathematical structures in computer science}\ }\textbf {\bibinfo {volume}
  {17}},\ \bibinfo {pages} {1169} (\bibinfo {year} {2007})}\BibitemShut
  {NoStop}%
\bibitem [{\citenamefont {Wong}(2020)}]{wong2020unstructured}%
  \BibitemOpen
  \bibfield  {author} {\bibinfo {author} {\bibfnamefont {T.~G.}\ \bibnamefont
  {Wong}},\ }\href@noop {} {\bibfield  {journal} {\bibinfo  {journal} {arXiv
  preprint arXiv:2011.14533}\ } (\bibinfo {year} {2020})}\BibitemShut {NoStop}%
\bibitem [{\citenamefont {Grover}(1996)}]{grover1996fast}%
  \BibitemOpen
  \bibfield  {author} {\bibinfo {author} {\bibfnamefont {L.~K.}\ \bibnamefont
  {Grover}},\ }in\ \href@noop {} {\emph {\bibinfo {booktitle} {Proceedings of
  the twenty-eighth annual ACM symposium on Theory of computing}}}\ (\bibinfo
  {year} {1996})\ pp.\ \bibinfo {pages} {212--219}\BibitemShut {NoStop}%
\bibitem [{\citenamefont {Fedichkin}\ \emph {et~al.}(2005)\citenamefont
  {Fedichkin}, \citenamefont {Solenov},\ and\ \citenamefont
  {Tamon}}]{fedichkin2005mixing}%
  \BibitemOpen
  \bibfield  {author} {\bibinfo {author} {\bibfnamefont {L.}~\bibnamefont
  {Fedichkin}}, \bibinfo {author} {\bibfnamefont {D.}~\bibnamefont {Solenov}},
  \ and\ \bibinfo {author} {\bibfnamefont {C.}~\bibnamefont {Tamon}},\
  }\href@noop {} {\bibfield  {journal} {\bibinfo  {journal} {arXiv preprint
  quant-ph/0509163}\ } (\bibinfo {year} {2005})}\BibitemShut {NoStop}%
\bibitem [{\citenamefont {Abramowitz}\ and\ \citenamefont
  {Stegun}(1964)}]{abramowitz+stegun}%
  \BibitemOpen
  \bibfield  {author} {\bibinfo {author} {\bibfnamefont {M.}~\bibnamefont
  {Abramowitz}}\ and\ \bibinfo {author} {\bibfnamefont {I.~A.}\ \bibnamefont
  {Stegun}},\ }\href@noop {} {\emph {\bibinfo {title} {Handbook of Mathematical
  Functions with Formulas, Graphs, and Mathematical Tables}}},\ \bibinfo
  {edition} {ninth dover printing, tenth gpo printing}\ ed.\ (\bibinfo
  {publisher} {Dover},\ \bibinfo {address} {New York},\ \bibinfo {year}
  {1964})\BibitemShut {NoStop}%
\bibitem [{\citenamefont {Nielsen}\ and\ \citenamefont
  {Chuang}(2002)}]{nielsenquantum}%
  \BibitemOpen
  \bibfield  {author} {\bibinfo {author} {\bibfnamefont {M.~A.}\ \bibnamefont
  {Nielsen}}\ and\ \bibinfo {author} {\bibfnamefont {I.}~\bibnamefont
  {Chuang}},\ }\href@noop {} {\enquote {\bibinfo {title} {Quantum computation
  and quantum information},}\ } (\bibinfo {year} {2002})\BibitemShut {NoStop}%
\end{thebibliography}%

\onecolumngrid

\newpage
\appendix

\section{{\em C\MakeLowercase{lassical random walk on $\mathcal{L}$}}}\label{classical}
\begin{center}
The proof of {\bf{}{Theorem \ref{thm1}}} is given below.    
\end{center}
 
\begin{proof}
Without loss of generality, let $n_1$ be the maximum of all $n_i$ for $1 \leq i \leq d$. We construct a coupling $(X_t, Y_t)$ of two walkers to  conduct a lazy random walk on $\mathbb{Z}_{n_1} \times \mathbb{Z}_{n_2} \times \dots \times \mathbb{Z}_{n_d}$. Assume that $X_t$ began at $x$ and $Y_t$ began at $y$. To couple these random walkers, we pick one of the $d$ coordinates at random, and if the two walkers agree there, we shift them both by $+1$, $-1$, or $0$ in that coordinate, with probabilities of $\frac{1}{4}$, $\frac{1}{4}$, and $\frac{1}{2}$, respectively. If they vary in the chosen coordinate, we fix one and move the other $+1$ or $-1$ in that coordinate, with the sign determined by a fair coin flip. Consider $X_t = (X^1_t,\dots, X^d_t)$ and $Y_t = (Y^1_t,\dots,Y^d_t)$, with $\tau_i = min\{t \geq 0 \vert X^i_t = Y^i_t\}$.

We can see that the anticipated number of moves to get $X^i_t$ and $Y^i_t$ to agree on a particular coordinate $i$ is at most $\frac{n_i^2}{4}$. There is a geometric waiting time between moves with $d$ expectation value because the chance of choosing the coordinate $i$ is $1/d$ at each step. This provides us with

\begin{equation}
    E_{x,y}(\tau_i) \leq \frac{d n_{i}^2}{4}.
\end{equation}

For $1\leq i \leq d$, and \[\tau_{couple} = \max_{1\leq i\leq d}\tau_i\]. 
Now the probability that the $i^{th}$ coordinates of two walkers have not yet coupled by time $t$ is less than $\Big(\frac{1}{4}\Big)^{t/n_i^2 d}$ for $1\leq i \leq d$.
So the probability that in none of the coordinates, these two walkers couple by time $t$  is less than \[\sum_{i =1}^{d}\Big(\frac{1}{4}\Big)^{t/n_i^2 d}.\]

Let's assume that the mixing time is $t = 2 n_1^2 d \log{(d/\epsilon)}$ then 

\[\sum_{i =1}^{d}\Big(\frac{1}{4}\Big)^{t/n_i^2 d}  < \epsilon,\]
must be true.  Take $c = \max\{ n_1/n_i |2 \leq i \leq d\}$,

\begin{equation}
    \begin{split}
\sum_{i =1}^{d}\Big(\frac{1}{4}\Big)^{t/n_i^2 d} &= \sum_{i =1}^{d}\Big(\frac{1}{4}\Big)^{\frac{2 n_1^2 d \log{(d/\epsilon)}}{n_i^2 d}}   \\
& \leq \sum_{i =1}^{d}\Big(\frac{1}{4}\Big)^{2 c^2  \log{(d/\epsilon)}}\\
& =\sum_{i =1}^{d}\Big(\frac{1}{16}\Big)^{c^2  \log{(d/\epsilon)}}\\
&  \leq d\Big(\frac{1}{16}\Big)^{c  \log{(d/\epsilon)}}\\
&  = d\Big((\epsilon/d)^{\log{(16)}}\Big)^c\\
&  = \epsilon^{c\log{(16)}} \frac{d}{d^{(c\log{(16)})}}\\
& \leq \epsilon.
\end{split}
\end{equation}

Hence, the proof.
\end{proof}

\section{{\em Q\MakeLowercase{uantum walk probability analysis on $\mathbb{Z}_{n}$}}}\label{zn}

  Richter's analysis of $\mathbb{Z}_{n}$, as presented in~\cite{PR1}, utilizes the asymptotic properties of Bessel functions to establish a linear mixing time that is proportional to the number of vertices. This is achieved by mapping the infinite line onto an $ n-$ cycle. Richter extends the mixing time result on the infinite line due to Childs in (Ref.~\cite{childs2002example}) to a finite cycle.
  
  Using the adjacency matrix, we analyze the behaviour of the continuous time quantum walk on $\mathbb{Z}_{n}$, i.e. $n$-cycle. We then obtain an expression for the probability of transitioning from one vertex to another at time $t$ and break it into time-dependent and time-independent parts. This analysis helps to prove the mixing time on the higher-dimensional case $\mathcal{L}$.

The adjacency matrix of $\mathbb{Z}_n$ is 

\begin{equation}\label{Eq;Adj}
    A = 
    W_{n} + W_{n}^{n-1}, 
\end{equation}

where $W_{n}$ is $n \times n$ a circulant matrix.

 \begin{equation*}
     W_{n} = \begin{bmatrix}
          0 & 1 & 0 & \dots & 0 \\
          0 & 0 & 1 & \dots & 0 \\
        \vdots & \vdots & \vdots & \vdots & \vdots \\  
         1 & 0 & 0 & \dots & 0
          \end{bmatrix},
 \end{equation*}
 with eigenvalues $w^j$ and eigenvectors 
 $v_j^{T} = (1/\sqrt{n})[1, w^j, w^{2j}, \dots, w^{(n-1)j}]$ for $0 \leq j \leq n-1$, where $w = e^{(2 \pi i/n)}$. The normalised adjacency matrix $\Bar{A} = A/2$ is the transition matrix of the classical random walk on an $n$-cycle. The $n$ eigenvalues for $j=1,2, \ldots, n-1$ are
 \begin{align*}
  \lambda_j & = \frac{1}{2}( w^j + w^{(n-1)j}) =  \cos{\left(\frac{2 \pi j}{n}\right)}.
      \end{align*}

The continuous time quantum walk operator $U(t) = e^{i\Bar{A}t}$ can been written as
  \begin{equation}
     U(t) = \sum_{j = 0}^{n-1} e^{i \lambda_j t} \ketbra{v_j}{v_j} .
  \end{equation}
  where \[\ketbra{v_j}{v_j} = \frac{1}{n}\begin{bmatrix}
          1 & \Bar{w}^j & \dots & \Bar{w}^{(n-1)j}  \\
          w^j & 1 &  \dots & \Bar{w}^{(n-2)j} \\
        \vdots & \vdots & \vdots & \vdots \\  
         w^{(n-1)j} & w^{(n-2)j} & \dots & 1
          \end{bmatrix}. \]

We can rewrite the $U(t)$ as 

\begin{equation}
     U(t) = \left(\sum_{j = 0}^{n-1} e^{it(\cos{(2 \pi j/n)})}\ketbra{v_j}{v_j} \right).
  \end{equation}

The probability to go from some vertex $\ket{p}$ to another vertex $\ket{q}$ on the graph $\mathbb{Z}_{n}$  in time $t$ is given by
\begin{equation}
    P_{t}(p,q) =|\bra{q}U(t)\ket{p}|^2 =\Bigg|\bra{q}\sum_{j = 0}^{n-1} e^{it(\cos{(2 \pi j/n)})} \ketbra{v_j}{v_j}  \ket{p}\Bigg|^2.
\end{equation}

We break the probability equation into time-dependent and time-independent parts.

\begin{equation}\label{eq:mod}
\begin{aligned}
 |.|^2 & = \frac{1}{n^2 }\left( n + \sum_{j,k = 0,j \neq k}^{n-1} e^{it(\cos{(2 \pi j/n)}-\cos{(2 \pi k/n)})} w^{(q-p)(j-k)}\right),\\
 & = \Bigg(\frac{1}{n}+ \frac{1}{n^2}\sum_{j = 1}^{n-1} w^{2(q-p)j} + \\
 & = \frac{1}{n^2} \sum_{j,k = 0,j \neq k, j+k \neq n}^{n-1} e^{it(-2\sin{(\pi (j+k)/n)}\sin{(\pi (j-k)/n)})} w^{(q-p)(j-k)}\Bigg).
\end{aligned}  
\end{equation}

To reduce the complexity, we define

\begin{align}
 \sum_{j = 1}^{n-1} w^{2(q-p)j} &=\begin{cases}
  n-1  & q-p = 0, \\
   -1  & \text{otherwise},   
\end{cases}  \\
& = n{\cal{I}}\{q-p=0\}-1,
\end{align}
where ${\cal{I}}$ is indicator function, and 

\begin{equation}\label{nt}
   n(t) = \sum_{j,k = 0,j \neq k,j+k \neq n}^{n-1} e^{-it(\sin{(\pi (j+k)/n)}\sin{(\pi (j-k)/n)})} w^{(q-p)(j-k)}
\end{equation}

 We will use this formulation in the mixing time on $\mathbb{Z}_{n_1}\times\mathbb{Z}_{n_2}$ proof, section~\ref{repeatedq}.

\section{{\em P\MakeLowercase{roof of lemma \ref{lemma2}}}}\label{lemma2app}

\begin{proof} 
We will establish (\ref{case:n1}); a similar argument justifies (\ref{case:n2}). The integral $\int_{0}^{T} n_1(t) dt$ can be 
computed directly by standard integration. (In the following we have replaced $j_1 - k_1$ by $- (k_1 -j_1)$ in Eq.~(\ref{nt}).)
\begin{equation}
 \begin{aligned}
 &  \left |\int_{0}^{T} n_1(t) dt \right| = \\     
& \Bigg| \sum_{j_1,k_1 = 0, k_1>j_1 , j_1+k_1 \neq n_1}^{n_1-1}
2\Bigg(\frac{\sin{\left(T\sin{(\pi (j_1+k_1)/n_1)}\sin{(\pi (k_1-j_1)/n_1)}+2\pi(l_1-0)(k_1-j_1)/n_1\right)}}{\sin{(\pi (j_1+k_1)/n_1)}\sin{(\pi (k_1-j_1)/n_1)}}\\
 &- \frac{\sin{(2\pi(l_1-0)(k_1-j_1)/n_1})}{\sin{(\pi (j_1+k_1)/n_1)}\sin{(\pi (k_1-j_1)/n_1)}}\Bigg)\Bigg|.
\end{aligned}   
\end{equation}  
{Taking absolute values, replacing the numerator by $1$, we obtain}
 \begin{align}\label{positive_term}
     & \leq 4\sum_{j_1,k_1 = 0, k_1>j_1 , j_1+k_1 \neq n_1}^{n_1-1}
\Bigg|\frac{1}{\sin{(\pi (j_1+k_1)/n_1)}\sin{(\pi (k_1-j_1)/n_1)}}\Bigg|.
 \end{align}
Note that the map $(j_1,k_1) \mapsto (j_1+k_1, k_1-j_1)$ from 
$\{(j_1,k_1): 0 \leq j_1 < k_1 \leq n_1-1\}$ to
$\{1,2,\dots, 2n_1 -3\}\times \{1,2, \dots, n_1-1\}$ (its inverse is
$(y,z) \mapsto ((y-z)/2, (y+z)/2)$).

\begin{align}
    &\leq 4\sum_{y=1, y\neq n_1}^{2n_1-3}\sum_{z =1 }^{n_1-1}
\Bigg|\frac{1}{\sin{(\pi (y)/n_1)}\sin{(\pi (z)/n_1)}}\Bigg|,\\
&\leq 4\sum_{y=1, y\neq n_1}^{2n_1-3}\sum_{z =1 }^{n_1-1}
\Bigg|\frac{1}{\sin{(\pi (y)/n_1)}\sin{(\pi (z)/n_1)}}\Bigg|,
\end{align}
\begin{align}   
&\leq \Bigg(4\sum_{y=1, y\neq n_1}^{2n_1-3}\Bigg|\frac{1}{\sin{(\pi (y)/n_1)}}
\Bigg|\Bigg)\Bigg(
\sum_{z =1 }^{n_1-1}
\frac{1}{\sin{(\pi (z)/n_1)}}\Bigg),\\
&\leq 8\Bigg(\sum_{y=1}^{n_1-1}\frac{1}{\sin{(\pi (y)/n_1)}}
\Bigg)\Bigg(
\sum_{z =1 }^{n_1-1}
\frac{1}{\sin{(\pi (z)/n_1)}}\Bigg). \label{eq:twosums}
\end{align}

We can bound the two sums in the RHS as follows:
\begin{equation}
    \begin{aligned}
     \sum_{y = 1}^{n_1-1}\frac{1}{\sin{(\pi y/n_1)}} &\leq  \sum_{y \leq (n_1-1)/1000} \frac{2}{\sin{(\pi y/n_1)}} +\sum_{y \geq (n_1-1)/1000, y < (n_1-1)/2} \frac{2}{\sin{(\pi y/n_1)}},\\
    & \leq \sum_{y \leq (n_1-1)/1000} \frac{2n_1}{\pi y} +  \frac{2}{\sin{(\pi/1000)}} \frac{n_1}{2}, \\
    & \leq \frac{2 n_1 \log{((n_1-1)/1000)}}{\pi}+ \frac{2}{\sin{(\pi/1000)}} \frac{n_1}{2},\\
    & \leq 2(n_1 \log{(n_1)}).
\end{aligned}
\end{equation}
We conclude from Eq.(~\ref{eq:twosums}) that
\begin{equation}
    \left|\int_{0}^{T} n_1(t) \right|dt \leq 32(n_1 \log(n_1))^2.
\end{equation}

\end{proof}

\section{{\em Q\MakeLowercase{uantum walk algorithm}}}\label{qalgo}

The quantum walk we consider is the following. 
The underlying graph 
$G(V,E)$ is the $d$-regular graph with vertex set $V=\mathbb{Z}_{n_1}\times \mathbb{Z}_{n_2}\times \cdots \times \mathbb{Z}_{n_d}$ and edge set $E$ consisting of pairs $\{x,x'\}$, where $\sum_{i} |x_i-x'_i| = 1$. 
Let $\Bar{A}$ be the normalised adjacency matrix of $G$ defined as $\Bar{A}(i,j) = 1/d$ if vertex $i$ is adjacent to vertex $j$ in $G$, and $0$ otherwise. Then the continuous time quantum walk operator for time $t$ is $U(t) = e^{i \Bar{A} t}$. Starting from state $\ket{x_0} = \ket{p}$, where $p \in V $, the quantum walk algorithm performs the following $T T'$ steps.

\bigskip
\noindent\hrule
\begin{center}
\textbf{Algorithm 1}
\end{center}
\noindent\hrule
\begin{itemize}
    \item Quantum walk algorithm $(p, T',T)$
    \begin{enumerate}
        \item $r = 0$; $\ket{x_0}=\ket{p}$;
        \item While $(T' \geq r)$ 
        \begin{itemize}
            \item Perform the quantum walk starting with $\ket{x_r}$ for time $t$ chosen uniformly at random from $[0,T]$;
            
        \item Let $\ket{\psi_{r+1}}= e^{i\Bar{A}t} \ket{x_r}$;
        \item Measure $\ket{\psi_{r+1}}$ in the position basis and obtain the state $\ket{x_{r+1}}$;
        \item $r = r+1$ ;
        \end{itemize}
    \item Output $\ket{x_r}$
    \end{enumerate}
\end{itemize}
\hrule
\bigskip
\noindent 

\section{{\em P\MakeLowercase{roofs of bounds in Theorem \ref{thm3}}}}\label{bounds}

We give here the bounds proved for the cases in Theorem $\ref{thm3}$.

\begin{center}

 \textbf{ case 0: $l_1 = 0$, $l_2 = 0$ }  
 
\end{center}

\begin{equation}
 \begin{aligned}
 \left|P_{T}(0,0) - \frac{1}{n_1n_2}\right|= & \Bigg|\frac{1}{T} \int_{0}^{T} \frac{1}{(n_1n_2)^2} \Bigg( n_1 (n_2-1) + n_1 n_2(t) + n_2 (n_1-1)+ (n_1-1)(n_2-1)\\
  &+ (n_1-1) n_2(t) + n_2 n_1(t)+ (n_2-1) n_1(t)  + n_1(t) n_2(t)  \Bigg)dt\Bigg|.  
\end{aligned}   
\end{equation}

\begin{equation}
 \begin{aligned}
  =& \Bigg|\frac{n_1 (n_2-1)}{(n_1n_2)^2} +  \frac{n_1}{(n_1n_2)^2} \frac{1}{T} \int_{0}^{T} n_2(t) dt + \frac{n_2 (n_1-1)}{(n_1n_2)^2}+ \frac{(n_2-1)(n_1-1)}{(n_1n_2)^2}\\
  &+\frac{n_1-1}{(n_1n_2)^2} \frac{1}{T} \int_{0}^{T} n_2(t) dt + \frac{n_2}{(n_1n_2)^2} \frac{1}{T} \int_{0}^{T} n_1(t) dt+\frac{n_2-1}{(n_1n_2)^2} \frac{1}{T} \int_{0}^{T} n_1(t) dt \\ 
  &+ \frac{1}{(n_1n_2)^2} \frac{1}{T} \int_{0}^{T} n_1(t) n_2(t) dt\Bigg|. 
\end{aligned}   
\end{equation}

\begin{equation}
 \begin{aligned}
 \left|P_{T}(0,0) - \frac{1}{n_1n_2}\right| & \leq \frac{n_1 (n_2-1)+ n_2 (n_1-1)+(n_1-1)(n_2-1)}{(n_1n_2)^2} \\
 &+ \frac{32(2n_1-1) (n_2 \log{n_2})^2 +32 (2n_2-1) (n_1 \log{n_1})^2}{T(n_1n_2)^2}\\
 & + \frac{ 32 n_1(n_2 \log(n_2))^2+ 32 n_2(n_1 \log(n_1))^2}{T(n_1n_2)^2},\\
 &\leq \frac{3}{n_1n_2} + \frac{64(\log(n_2))^2}{n_1 T} + \frac{64(\log(n_1))^2}{n_2 T} + \frac{32 n_2(\log(n_2))^2+ 32n_1(\log(n_1))^2}{(n_1n_2) T}.
\end{aligned}   
\end{equation}

For $T = 1600(n_1 +n_2)(\log{(n_1)^2})$,

\begin{equation}
 \begin{aligned}
\left|P_{T}(0,0) - \frac{1}{n_1n_2}\right| & \leq    \frac{3}{n_1n_2} + \frac{64(\log(n_2))^2}{n_1 1600(n_1 +n_2)(\log{(n_1)^2})} + \frac{64(\log(n_1))^2}{n_2 1600(n_1 +n_2)(\log{(n_1)^2})}, \\
&+ \frac{32 n_2(\log(n_2))^2+ 32 n_1(\log(n_1))^2}{(n_1n_2)1600(n_1 +n_2)(\log{(n_1)^2})},\\
& \leq \frac{3}{n_1n_2} + \frac{1}{25n_2^2} + \frac{1}{50n_2^2}\\
& \leq \frac{4}{ n_2^2}.
\end{aligned}   
\end{equation}

\begin{center}
 
 \textbf{case 1 : $l_1 = 0$ and $l_2 \neq 0$} (There are $(n_2 -1)$ such terms.)  
 
\end{center}

\begin{equation}
 \begin{aligned}
  \left|P_T(0,l) - \frac{1}{n_1n_2}\right| & = \Bigg|\frac{1}{T} \int_{0}^{T} \frac{1}{(n_1n_2)^2} \Bigg( - n_1 + n_1 n_2(t) + n_2 (n_1 -1)\\
    &-(n_1 - 1)+(n_1 -1) n_2(t) + n_2 n_1(t)\\
    &- n_1(t)  + n_1(t) n_2(t)  \Bigg)dt\Bigg|.
\end{aligned}   
\end{equation}

\begin{equation}
    \begin{aligned}
     &\leq \frac{2 n_1 n_2}{(n_1n_2)^2}+ \frac{32 (2n_1-1)(n_2\log{n_2})^2}{T(n_1n_2)^2}  \\
     &+ \frac{32 (n_2-1)(n_1\log{n_1})^2}{T(n_1n_2)^2} + \frac{32 n_2(\log(n_2))^2+ 32n_1(\log(n_1))^2}{(n_1n_2) T},\\
     & \leq \frac{2}{n_1 n_2} + \frac{64(\log{n_2})^2}{n_1 T} + \frac{32(\log{n_1})^2}{n_2T}+ \frac{32 n_2(\log(n_2))^2+ 32n_1(\log(n_1))^2}{(n_1n_2) T}.
\end{aligned}
\end{equation}

We multiply by $n_2$ to get bound on these terms. 

\begin{equation}
    \begin{aligned}
   n_2\left|P_T(0,l) - \frac{1}{n_1n_2}\right| & \leq  \frac{2}{n_1} + \frac{64 n_2(\log{n_2})^2}{n_1 T} + \frac{32(\log{n_1})^2}{T}\\
   & + \frac{32 n_2(\log(n_2))^2+ 32n_1(\log(n_1))^2}{n_1 T}.   
\end{aligned}
\end{equation}

 So for $T = 1600(n_1+n_2)(\log{n_1})^2$,

\begin{equation}
    \begin{aligned}
   n_2\left|P_T(0,l) - \frac{1}{n_1n_2}\right| & \leq  \frac{2}{n_1} + \frac{1}{50 n_2} + \frac{1}{100 n_2}+\frac{1}{50 n_2},\\
   & = \frac{205}{100 n_2},\\
   & \leq \frac{3}{n_2}.
\end{aligned}
\end{equation}

\begin{center}
 
 \textbf{case 2 : $l_1 \neq 0$ and $l_2 = 0$} (There are $(n_1 -1)$ such terms.)
   
\end{center}

\begin{equation}
 \begin{aligned}
  \left|P_T(0,l) - \frac{1}{n_1n_2}\right| & = \Bigg|\frac{1}{T} \int_{0}^{T} \frac{1}{(n_1n_2)^2} \Bigg( n_1(n_2-1) + n_1 n_2(t) - n_2\\
    &-(n_2 - 1)- n_2(t) + n_2 n_1(t)\\
    &+(n_2 -1)n_1(t)  + n_1(t) n_2(t)  \Bigg)dt\Bigg|.
\end{aligned}   
\end{equation}

\begin{equation}
\begin{aligned}
    \leq \frac{2}{n_1 n_2} + \frac{64(\log{n_1})^2}{n_2 T} + \frac{32(\log{n_2})^2}{n_1T}+ \frac{32 n_2(\log(n_2))^2+ 32 n_1(\log(n_1))^2}{(n_1n_2) T}. 
\end{aligned}    
\end{equation}

\begin{align*}
     n_1\left|P_T(0,l) - \frac{1}{n_1n_2}\right| & \leq 
     \frac{2}{n_2} + \frac{64 n_1(\log{n_1})^2}{n_2 T} + \frac{32(\log{n_2})^2}{T}\\
     &+ \frac{32 n_2(\log(n_2))^2+ 32 n_1(\log(n_1))^2}{n_2 T}.
\end{align*}

We multiply by $n_1$ to get bound on these terms. For $T = 1600 (n_1 + n_2)(\log{n_1})^2$,

\begin{equation}
 \begin{aligned}
  n_1\left|P_T(0,l) - \frac{1}{n_1n_2}\right| & \leq 
        \frac{2}{n_2} + \frac{1}{25 n_2} + \frac{1}{100 n_2} +\frac{1}{50 n_2},\\
        & \leq \frac{3}{n_2}.
\end{aligned}   
\end{equation}

\begin{center}

\textbf{case 3 : $l_1 \neq 0$ and $l_2 \neq 0$} (There are $(n_1 n_2 -n_1- n_2 + 1)$ such terms.) 
    
\end{center}

\begin{equation}
 \begin{aligned}
  \left|P_T(0,l) - \frac{1}{n_1n_2}\right| & = \Bigg|\frac{1}{T} \int_{0}^{T} \frac{1}{(n_1n_2)^2} \Bigg( -n_1 + n_1 n_2(t) - n_2 + 1 - n_2(t)\\ 
  &+ n_2 n_1(t)- n_1(t)  + n_1(t) n_2(t)  \Bigg)dt\Bigg|,\\
  & = \Bigg|\frac{1}{T} \int_{0}^{T} \frac{1}{(n_1n_2)^2} \Bigg( -(n_1 +n_2) + 1 - (n_2(t)+ n_1(t)) \\
  &+ n_1 n_2(t)+ n_2 n_1(t) + n_1(t) n_2(t)  \Bigg)dt\Bigg|,\\
  & \leq \frac{(n_1 +n_2+1)}{(n_1n_2)^2} +\frac{32(n_2 \log(n_2))^2}{T(n_1n_2)^2} + \frac{32(n_1 \log(n_1))^2}{T(n_1n_2)^2}\\
  &+ \frac{32 n_1(n_2 \log(n_2))^2}{T(n_1n_2)^2} + \frac{32 n_2(n_1 \log(n_1))^2}{T(n_1n_2)^2} \\
  &+ \frac{32 n_1(n_2 \log(n_2))^2+ 32 n_2(n_1 \log(n_1))^2}{T(n_1n_2)^2}. 
\end{aligned}   
\end{equation}
We will multiply it by $n_1 n_2$ to get the bound on these terms.

\begin{equation}
\begin{aligned}
  n_1 n_2 \left|P_T(0,l) - \frac{1}{n_1n_2}\right| & \leq \frac{(n_1 +n_2+1)}{(n_1n_2)} +\frac{32(n_2 \log(n_2))^2}{T(n_1n_2)} + \frac{32(n_1 \log(n_1))^2}{T(n_1n_2)}\\
  &+ \frac{32 n_1(n_2 \log(n_2))^2}{T(n_1n_2)} + \frac{32 n_2(n_1 \log(n_1))^2}{T(n_1n_2)} \\
  &+ \frac{32 n_1(n_2 \log(n_2))^2+ 32 n_2(n_1 \log(n_1))^2}{T(n_1n_2)}.   
\end{aligned}    
\end{equation}

For $T = 1600 (n_1 + n_2) (\log{(n_1)})^2$,

\begin{align*}
    n_1 n_2 \left|P_T(0,l) - \frac{1}{n_1n_2}\right| & \leq \frac{(n_1 +n_2+1)}{(n_1n_2)} +\frac{32(n_2 \log(n_2))^2}{1600 (n_1 + n_2) (\log{(n_1)})^2(n_1n_2)} \\
  &+ \frac{32(n_1 \log(n_1))^2}{1600 (n_1 + n_2) (\log{(n_1)})^2(n_1n_2)}+ \frac{32 n_1(n_2 \log(n_2))^2}{1600 (n_1 + n_2) (\log{(n_1)})^2(n_1n_2)}
\end{align*}
\begin{equation}
\begin{aligned}
  &+ \frac{32 n_2(n_1 \log(n_1))^2}{1600 (n_1 + n_2) (\log{(n_1)})^2(n_1n_2)} + \frac{32 n_1(n_2 \log(n_2))^2+ 32 n_2(n_1 \log(n_1))^2}{1600 (n_1 + n_2) (\log{(n_1)})^2(n_1n_2)},\\
  & \leq \frac{(n_1 +n_2+1)}{(n_1n_2)} + \frac{3}{100 n_2} + \frac{2}{50}, \\
  & \leq \frac{3}{n_2}+ \frac{2}{50}.
\end{aligned}    
\end{equation}

\section{\em S\MakeLowercase{imulation supporting  conjecture}}\label{simul}

We considered pairs $(n_1, n_2)$ such that $n_1 > n_2$ are odd and relatively prime, and in the range $[10, 200]$. There are 3685 such pairs in this range. We plotted the bound given by the right-hand side of Eq.(~\ref{d2}) in orange, and the value of $\left|\int_{0}^{T} n_1(t) n_2(t) dt\right|$ in blue, as a function of time $T$. The results are shown in Fig. \ref{fig:bound1}.

We observed that the bound holds for all pairs $(n_1, n_2)$, and the increasing gap between the two curves suggests that it will also hold for later times. These numerical simulations provide supporting evidence for our conjecture.

\begin{figure}[ht!]
        \centering
        \includegraphics[ trim = 110 240 100 235, width = \textwidth, keepaspectratio]{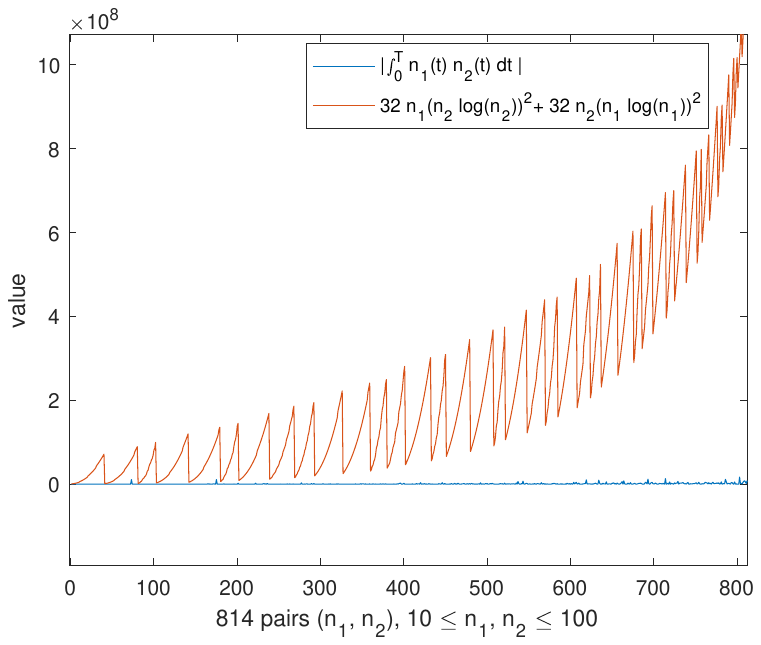}
        \caption{\textbf{Numerical evidence of the bound in Eq.~\ref{d2}:} In this figure, we plot LHS(\textcolor{blue}{blue} blue colour) of conjecture \ref{conj1} for different $(n_1, n_2)$ pairs where $10 \leq n_1, n_2 \leq 100$. The plot looks fluctuating because in numerical simulation, once $n_2$ is selected, then $n_1$ takes all values between $(n_2+1, 100)$. The bound is given in \textcolor{orange}{orange} colour. Note that the plot is scaled at $10^8$ order; hence, the \textcolor{blue}{blue} line seems flat.}
        \label{fig:bound1}
    
    \end{figure}

\end{document}